\documentclass[10pt]{article}
\usepackage[utf8]{inputenc}
\usepackage[english]{babel}
\usepackage{hyperref}

\usepackage{amsthm}

\newtheorem{lem}{Lemma}
\usepackage{graphicx}
\usepackage{stmaryrd}
\usepackage{mathrsfs}
\usepackage{amssymb, amsmath, amsfonts, amsthm}
\usepackage{url}
\usepackage{wasysym}
\usepackage{tikz}
\usepackage{subcaption}
\usepackage{authblk}

\usetikzlibrary{positioning} 
\usetikzlibrary{shapes} 
\usetikzlibrary{arrows,decorations.markings} 
\tikzstyle directed=[postaction={decorate,decoration={markings,mark=at position .65 with {\arrow[scale=1.2]{>}}}}]

\usepackage{color} 
\usepackage{ulem} 
\newcommand{\erase}[1]{}

\renewcommand{\bar}{\overline}

\newcommand{\oo}[2]{\left\rrbracket #1,#2\right\llbracket}
\newcommand{\co}[2]{\left\llbracket #1,#2\right\llbracket}
\newcommand{\cc}[2]{\left\llbracket #1,#2\right\rrbracket}
\newcommand{\ie}{\textit{i.e. }}

\newcommand\compl[1]{\overline{#1}}
\usepackage[nottoc, notlof, notlot]{tocbibind}

\DeclareMathOperator{\GNECC}{\mathsf{NECC}}
\DeclareMathOperator{\GINECC}{\mathsf{INECC}}
\DeclareMathOperator{\NECCs}{NECCs}
\DeclareMathOperator{\NECC}{NECC}
\DeclareMathOperator{\INECC}{INECC}
\DeclareMathOperator{\COLOR}{COLOR}

\DeclareMathOperator{\CC}{CC}
\DeclareMathOperator{\NEC}{NEC}

\newtheorem{remark}{Remark}
\newtheorem{definition}{Definition}
\newtheorem{theorem}{Theorem}
\newtheorem{conjecture}{Conjecture}

\begin{document}
	
\title{On the cost of simulating a parallel Boolean automata network by a block-sequential one}

\author[1]{Florian Bridoux}
\author[2]{Pierre Guillon}
\author[1]{K\'evin Perrot}
\author[1,3]{Sylvain Sen\'e}
\author[2]{Guillaume Theyssier}

\affil[1]{Universit\'e d'Aix-Marseille, CNRS, LIF, Marseille, France}
\affil[2]{Universit\'e d'Aix-Marseille, CNRS, Centrale Marseille, I2M, 
	Marseille, France}
\affil[3]{Institut rhône-alpin des systèmes complexes, IXXI, Lyon, France}

\maketitle

\begin{abstract}
In this article we study the minimum number $\kappa$ of additional automata that a Boolean automata network (BAN) associated with a given block-sequential update schedule needs in order to simulate a given BAN with a parallel update schedule. We introduce a graph that we call $\GNECC$ graph built from the BAN and the update schedule. We show the relation between $\kappa$ and the chromatic number of the $\GNECC$ graph. Thanks to this $\GNECC$ graph, we bound $\kappa$ in the worst case between $n/2$ and $2n/3+2$ ($n$ being the size of the BAN simulated) and we conjecture that this number equals $n/2$. We support this conjecture with two results: the clique number of a $\GNECC$ graph is always less than or equal to $n/2$ and, for the subclass of bijective BANs, $\kappa$ is always less than or equal to $n/2+1$.
\end{abstract}


\section{Introduction}\label{intro}

In this article, we study Boolean automata networks (BANs). A BAN can be seen as a set of two-states automata interacting with each other and evolving in a discrete time. BANs have been first introduced by McCulloch and Pitts in the $1940^s$~\cite{McCulloch1943}. They are common representational models for natural dynamical systems like neural or genetic networks~\cite{Goles1990,Hopfield1982,Kauffman1969,Kauffman1971,Thomas1973}, but they are also computational models with which we can study computability or complexity. In this article we are interested in intrinsic simulations between BANs, \ie simulations that focus on the dynamics rather than the computational power. More concretely, given a BAN $A$ we want to find a BAN $B$ which reproduces the dynamics of A while it satisfies some constraints. There have been few studies using intrinsic simulation between BANS before the $2010^s$~\cite{Bruck1988,Goles1993,Orponen1997,Tchuente1986}. More recently, this notion has received a new interest~\cite{Melliti2013,Melliti2016,Noual2012,Noual2013} and we are convinced that it is essential and deserves to be dealt with. Meanwhile, intrinsic simulation of many other similar objects (cellular automata, tilings, subshifts, self-assembly, etc.) has been really developing since $2000$~\cite{Delorme2011a,Delorme2011b,Doty2012,Guillon2011,Lafitte2007,Lafitte2009,Ollinger2012}.

A given BAN can be associated with several dynamics, depending on the schedule (\ie the order) chosen to update the automata. In this article, we will consider all block-sequential update schedules: we group automata into blocks, and we update all automata of a block at once, and iterate the blocks sequentially. Among these update schedules are the following classical ones: the parallel one (a unique block composed of $n$ automata) and the $n!$ sequential ones ($n$ blocks of $1$ automaton). The pair of a BAN and its update schedule is called a scheduled Boolean automata network (SBAN). 

For the last 10 years, people have studied the influence of the update schedules on the dynamics of a BAN~\cite{Aracena2009,Demongeot2008,Goles2012,Goles2008}. Here, we do the opposite. We take a SBAN, and try to find the smallest SBAN with a constrained update schedule which simulates this dynamics. For example, let $N$ be a parallel SBAN of size 2 with 2 automata that exchange their values. There are no SBANs $N'$ of size $2$ with a sequential update schedule which simulates $N$. Indeed, when we update the first automaton, we necessarily erase its previous value. If we did not previously save it, we cannot use the value of the first automaton to update the second automaton. Thus, $N'$ needs an additional automaton to simulate $N$ under the sequential update schedule constraint. A SBAN $N$ of size $n$ with a parallel update schedule can always be simulated by a SBAN $N'$ of size $2n$ with a given sequential update schedule. Indeed, we just need to add $n$ automata which copy all the information from the original automata and then, we compute sequentially the updates of the originals automata using the saved information. The goal of this article is to establish more precise bounds on the number of required additional automata, function of $n$, in the worst case.\medskip

In Section~\ref{section_def}, we define BANs and detail the notion of simulation that we use.
In Section~\ref{section_added_automata}, we consider the dynamics of a BAN $F$ with automata set $V$ and the parallel update schedule and we consider a block-sequential update schedule $W$. We focus on the minimum number $\kappa(F,W)$ of additional automata that a SBAN needs to simulate this dynamics with an update schedule identical to $W$ on $V$.
In Section~\ref{section_necc}, we define a graph which connects configurations depending on a BAN $F$ and a block-sequential update schedule $W$. We prove that the chromatic number of this graph determines the number $\kappa(F,W)$ defined in the previous section. We also state the following conjecture: $\kappa(F,W)$ is always less than or equal to $n/2$, where $n$ is the size of the BAN $F$.
In Section~\ref{section_inecc}, we define another graph constructed from the previous graph where we identify configurations which have the same image. We prove that the chromatic number of this new graph is always greater than that of the previous graph. We deduce an upper bound for $\kappa(F,W)$.
In Section~\ref{section_clique_number}, we try to support our conjecture by finding an upper bound for the clique number of the graph defined in Section~\ref{section_necc}.
Finally, in Section~\ref{section_special_classes}, we study $\kappa(F,W)$ in the case where F is bijective.

\section{Definitions and notations} \label{section_def} 

\subsection{BANs and SBANs}

In this article, unless otherwise stated, BANs have a size $n \in \mathbb{N}$, which means that they are composed of $n$ automata numbered from $0$ to $n-1$. Usually, we denote this set of automata by $V= \{ 0, 1, \dots, n-1 \}$ (which will be abbreviated by $\co0n$). Each automaton can take two states in the Boolean set $\mathbb{B} = \{ 0, 1 \}$.
A \emph{configuration} is a Boolean vector of size $n$, interpreted as the sequence of states of the automata of the BAN. In other words, if $x$ is a configuration, then $x \in \mathbb{B}^n$ and $x = (x_0, \dots, x_{n-1})$ with $x_i$ the state of automaton $i$ (for all $i$ in $V$).
For all $I\subseteq V$, we denote by $x_I$ the restriction of $x$ to $I$. In other words, if $I = \{i_1, i_2, \dots, i_p\}$ with $i_1 < i_2 < \dots < i_p$ then $x_I = (x_{i_1}, x_{i_2}, \dots, x_{i_p})$. We also denote by $x_{\bar{I}}$ the restriction of $x$ to $V \setminus I$.

For all $b \in \mathbb{B}$, we denote by $\bar{b}$ the negation of the state of $b$. In other words, $\bar{0} = 1$ and $\bar{1} = 0$.
We also denote by $\bar{x}$ the negation of $x$, such that $\bar{x} = (\overline{x_0}, \dots, \overline{x_{n-1}})$.
Furthermore, we denote by $\bar{x}^i$ or $\bar{x}^I$ the negation of $x$ respectively restricted to an automaton $i$ or a set $I$ of automata, that is, $\bar x^I_i=\bar{x_i}$ if $i\in I$, and $\bar x^I_i=x_i$ if $i\in V\setminus I$. 

In this article, we only study BANs with block-sequential update schedules. A SBAN $N = (F,W)$ is characterized by:
\begin{itemize}
\item a global update function $F:\mathbb{B}^{n} \rightarrow \mathbb{B}^{n}$ which 
	represents the BAN;
\item a block-sequential update schedule $W$.
\end{itemize}

The \emph{global update function} of a BAN is the collection of the local update functions of the BAN: we have $F(x) = (f_0(x), \dots, f_{n-1}(x))$, where for all $i \in V$, $f_i:\mathbb{B}^{n} \rightarrow \mathbb{B}$ is the local update function of automata $i$. We also use the $I$-update function $F_I$, with $I \subseteq V$, which gives a configuration where the states of automata in $I$ are updated and the other ones are not.
In other words, $\forall i \in V,\ F_I(x)_i = f_i(x)$ if $ i \in I $ and $x_i$ otherwise.
And, for singleton, we simply write $F_i(x) = F_{\{i\}}(x)$.

\begin{remark}
It is important not to confuse $F_I(x)$ and $F(x)_I$. The first one is the $I$-update function that we have just defined. The second is the configuration $F(x)$ restricted to $I$.
\end{remark}  

A \emph{block-sequential update schedule} is an ordered partition of $V$. The set of ordered partitions of $V$ is denoted by $\overrightarrow{\mathscr{P}}(V)$. Let $W \in \overrightarrow{\mathscr{P}}(V)$ and $p=|W|$ and ${W = (W_0,\ldots,W_{p-1})}$. We make particular use of $F^W$ defined as $F^W = F_{W_{p-1}} \circ \dots \circ F_{W_0}$. If $x \in \mathbb{B}^n$ is the configuration of the BAN at some time step, then $F^W(x)$ is the configuration of the BAN at the next step. There are two very particular kinds of block-sequential update schedules:
\begin{itemize}
\item the parallel update schedule where all automata are updated at the same time step. 
	So, we have $W = [V]$ (\ie $|W| = 1$ and $W_0 = V$) and $F^W = F$;
\item the sequential update schedules where automata are updated one at the time. So, we have $|W| = n$ and $\forall i \in \co0n,\ |W_i| = 1$.
\end{itemize}

For any $j \in \cc0p$, we denote $W_{<j} = \bigcup \limits_{i=0}^{j-1} W_i$. In particular, we have $W_{<0} = \emptyset$ and $W_{<p} = V$. Furthermore, for any $i \in \cc0p$, we denote $W^{<i} = (W_0, W_1, \dots, W_{i-1})$. $W^{<i}$ is an ordered partition of $W_{<i}$. In particular, we have $W^{<0} = [\ ]$ (the empty vector) and $W^{<p} = W$.

We will often use the following two notations:
\begin{enumerate}
\item[(i)] $F^{W^{<j}} = F_{W_{j-1}} \circ \dots \circ F_{W_0}$ is the function which makes the first $j$ steps of the transition of the SBAN $(F,W)$;
\item[(ii)] $F_{W_{<j}} = F_{W_0 \cup \dots \cup W_{j-1}}$ is the function which updates only the automata in the first $j$ blocks of 
$W$.
\end{enumerate}

Let $W \in \overrightarrow{\mathscr{P}}(V)$ be an update schedule. We know that each automaton of a block-sequential SBAN is updated only in one step of the update schedule. We denote by $W(i)$ the step at which $i$ is updated. More formally, $\forall i \in V, W(i)$ is the number $j \in \co0p$ such that $i \in W_j$.

\subsection{Simulation}

Here, we define the notion of simulation used in this article. We consider that a SBAN $N$ of size $m$ simulates another SBAN $N'$ of size $n$ if there is a projection from $\mathbb{B}^m$ to $\mathbb{B}^n$ such that the projection of the update in $N'$ equals the update in $N$ of the projection.

\begin{definition}
Let $F: \mathbb{B}^n \to \mathbb{B}^n$ and $F': \mathbb{B}^m \to \mathbb{B}^m$ with $m \geq n$, $V = \co0n$ and $V'= \co0m$, $W \in \overrightarrow{\mathscr{P}}(V)$ and $W' \in \overrightarrow{\mathscr{P}}(V')$. Let $h :  V \to V'$ be an injective function and $\varphi_h: \mathbb{B}^m \to \mathbb{B}^n$ be defined by $\varphi_h(x)=(x_{h(i)})_{i \in V}$.
We say that $(F',W')$ \emph{$h$-simulates} $(F,W)$, and note $(F',W') \vartriangleright^h (F,W)$, if $\varphi_h \circ F'^{W'} = F^{W} \circ \varphi_h$.
Moreover, $(F',W')$ \emph{simulates} $(F,W)$, which is denoted by $(F',W') \vartriangleright (F,W)$ if there is a $h$ such that $(F',W') \vartriangleright^h (F,W)$.
\end{definition}

In this article we often use an $id$-simulation which is a $h$-simulation with $h$ the identity function (${h(i)=i}$).

\section{Number of required additional automata } \label{section_added_automata} 

In this section, we define the main object of this article. Given a BAN $F$ with automata $V$ and a block-sequential update schedule $W\in\overrightarrow{\mathscr{P}}(V)$, we consider the smallest SBAN $(F',W')$ which simulates the parallel SBAN $(F,[V])$, where $W'$ extends $W$ by preserving its order. 
We could as well study the problem of finding a block-sequential SBAN $(F',W')$ which simulates another block-sequential SBAN $(G,W)$. However, this problem is in fact the same. Indeed, for any block-sequential SBAN $(G,W)$, the parallel SBAN $(G^W,[V])$ $id$-simulates $(G,W)$.

Let us formalize the notion.
From an update schedule $W$ and a BAN of size $n$, we define the notion of update schedule extending $W$ for a bigger BAN of size $m$. Let $V' = \co0m$. Let $h : V \to V'$ be an injective function. We denote by $\mathscr{E}_h(W,V') $ the set of update schedules $W'$ extending $W$ such that each $W'$ preserves the order of $W$ for the projection by $h$ of the automata of $V$.
That is to say, if one automaton is updated before another one according to $W$, then the projection of these automata into $V'$ will preserve the same update order in $W'$.
More formally, $\mathscr{E}_h(W,V') = \{ W' \in \overrightarrow{\mathscr{P}}(V')\ |\ \forall i \in V, 
W(i) \le W(i') \iff W'(h(i)) \le W'(h(i')) \}$.
In particular, if two automata $i,j\in V$ are updated at the same step $W(i)=W(j)$, then the projections $h(i),h(j)$ of these automata are updated at the same step $W(h(i))=W(h(j))$ in $W'$. In other words, $h$ induces a map $\tilde h:\co0p\to\co0{p'}$ such that $W(h(i))=\tilde h(W(i))$ for all $i\in V$, and $\tilde h(W)$ is a subordered partition of $W'$.

\begin{definition}
If $F$ is a BAN over automata $V=\co0n$ and $W\in\overrightarrow{\mathscr{P}}(V)$ is an update schedule, we define $\kappa(F,W)$ as the smallest $k$ such that there exist 
 an update schedule $W' \in \mathscr{E}_h(W,V')$ extending $W$ and a BAN $F': \mathbb{B}^{n+k} \to \mathbb{B}^{n+k}$ such that  $(F',W') \vartriangleright (F,[V])$, with $V'=\co0{n+k}$.

Furthermore, $\kappa_{n}$ is the value of $\kappa(F,W)$ in the worst case among all SBANs with automata $V$. In other words, $\kappa_{n} = max(\{ \kappa(F,W) \ |\ F : \mathbb{B}^{n} \to \mathbb{B}^{n} \text{ and } W \in \overrightarrow{\mathscr{P}}(V) \} )$.
\end{definition}

\section{$\NECCs$ set and $\GNECC$ graph} \label{section_necc} 

In order to answer the main problem of this article which is is to bound the values of $\kappa_n$, we introduce a new concept: the not equivalent and confusable configurations or $\NECCs$ and the $\GNECC$ graph. Theorem~\ref{theorem_chro} will show that the logarithm of the chromatic number of the $\GNECC$ graph of a SBAN and the $\kappa$ of this SBAN are equal. $\NEC$ (the acronym standing for \emph{non-equivalent configurations}) is the set of pairs of configurations with different images by $F$. In other words,
\begin{equation*} 
	\NEC_{F} = \{ (x,x') \in \mathbb{B}^n \times \mathbb{B}^n\ |\ F(x) \neq 
	F(x')\}\text{.}
\end{equation*}
We call \emph{confusable configurations} and denote by $\CC_{F,W}$, or simply $\CC$ (the acronym standing for \emph{confusable configurations}), the set of pairs of configurations which become identical when we update the first $i$ blocks of $W$ for some $i \in \co0p$). Formally, 
\begin{equation*}
	\CC = \{ (x,x') \in \mathbb{B}^n \times \mathbb{B}^n\ |\ \exists i \in \cc0p,\ F_{W_{<i}}(x) = F_{W_{<i}}(x')\}\text{.}
\end{equation*}

\begin{definition}
\emph{$\NECC_{F,W}$}, or simply $\NECC$ (the acronym standing for \emph{not equivalent and confusable configurations}), is the set of pairs of configurations which are confusable and not equivalent at the same time, $\NECC_{F,W} = \CC_{F,W} \cap \NEC_F$.
\end{definition}

\newcommand{\ICC}{\mathsf{CC}}
Also, for all $x, x' \in \mathbb{B}^n$, we denote by $\ICC_{F,W}(x, x')$ (or just $\ICC_{F,W}(x,x')$) the set of time steps $i$ which make them confusable. More formally, $\forall x,x' \in \mathbb{B}^n,\ \ICC_{F,W}(x,x') = \{i \in \cc0p\ |\ F_{W_{<i}}(x) = F_{W_{<i}}(x')\}$.

\begin{remark}
We have $\ICC(x,x') = \emptyset$ if and only if $(x,x') \not \in \CC$. 
\end{remark}

\begin{definition}
The \emph{$\GNECC$ graph}, denoted by $(\mathbb{B}^n,\NECC)$, is the nondirected graph which has the set of configurations $\mathbb{B}^n$ as nodes and the set of $\NECC$ pairs as edges.
\end{definition}

In the sequel, we make a particular use of two concepts of graph theory.
A \emph{valid coloring} of $G$ is a coloring of all the nodes of $G$ such that two adjacent nodes do not have the same color.
We denote by $\chi(G)$ the \emph{chromatic number} of the graph $G$, namely the minimum number of colors of a valid coloring of $G$. Furthermore, the chromatic number of the $\GNECC$ graph is denoted by $\chi(\GNECC) = \chi((\mathbb{B}^n,\NECC))$. We see in Lemma~\ref{lem1} that we can get a valid coloring of the $\GNECC_{F,W}$ graph from the SBAN $(F',W')$ which simulates $(F,[V])$. This coloring does not use more than $2^k$ colors with $k$ the number of additional automata of $F'$. We color the configuration of the $\GNECC$ graph using the values of the added automata after the update.

\begin{lem} 
\label{lem1}
For any BAN $F: \mathbb{B}^{n} \to \mathbb{B}^{n}$ and any block-sequential update schedule $W$, $\kappa(F,W) \geq \lceil \log_2(\chi(\GNECC_{F,W})) \rceil$.
\end{lem}

\begin{proof}
Let $h : V \to V'$ injective, $W' \in \mathscr{E}_h(W,V')$, $p = |W|$, $p' = |W'|$ and $F': \mathbb{B}^{n+k} \to \mathbb{B}^{n+k}$ such that $(F',W') \vartriangleright^{h} (F,[V])$.
We prove that $k \geq \lceil \log_2(\chi(\GNECC)) \rceil$. 
Let $z,z'$ be such that $z_{\compl{h(V)}} = z'_{\compl{h(V)}} = [0]^k$ and $(x,x')=(\varphi_h(z),\varphi_h(z')) \in \NECC$, and let us prove that $F'(z)_{\compl{h(V)}}\ne F'(z')_{\compl{h(V)}}$.
Suppose the contrary. Since $(x,x') \in \NECC$, we have $F(x) \ne F(x') $ and $\exists j \in \cc0p,$ $F_{W_{<j}}(x) = F_{W_{<j}}(x')$. 
Let $Z = F'^{{W'}^{<\tilde h(j)}}(z) =
$ and $Z' =F'^{{W'}^{<\tilde h(j)}}(z')
$. By assumption, we have $z_{\compl{h(V)}} = [0]^k = z'_{\compl{h(V)}}$ and $F(z)_{\compl{h(V)}} = F(z')_{\compl{h(V)}}$. Thus, $Z_{\compl{h(V)}} = Z'_{\compl{h(V)}}$.
Furthermore, we have $\varphi_h(Z) = F_{W_{<j}}(x) = F_{W_{<j}}(x') = \varphi_h(Z')$.
As a result, $Z_{h(V)} = Z'_{h(V)}$ and $Z = Z'$.
Consequently, $F'(z) = 
 F_{W'_{p-1}} \circ \dots \circ F_{W'_{\tilde h(j)}}(Z)$ and $F'(z') = 
 F_{W'_{p-1}} \circ \dots \circ F_{W'_{\tilde h(j)}}(Z') $ are equal. 
However, $(x,x') \in NEC$. Thus, $F'(z)_{h(V)} = F(x) \ne F(x') = F'(z')_{h(V)}$. As a consequence, we have also $F'(z) \ne F'(z')$. There is a contradiction.
We have proven that if $(x,x') \in \NECC$ then $F(z)_{\compl{h(V)}} \ne F(z')_{\compl{h(V)}}$.
In other words, a valid coloring of $\GNECC$ is obtained by coloring each vertex $x$ by $F(z)_{\compl{h(V)}}$, where $\phi_h(z)=x$ and $x_{\compl{h(V)}}=[0]^k$.
Hence $\{ F(z)_{\compl{h(V)}} | z_{\compl{h(V)}} = [0]^k  \}$ has at least $\chi(\GNECC)$ different values.
To encode these values, we need to have $k=|\compl{h(V)}| \geq \lceil  \log_2(\chi(\GNECC) \rceil $. So $\kappa(F,W) \geq \log_2(\chi(\GNECC))$.
\end{proof}

We see in Lemma~\ref{lem2} that we can get a SBAN $(F',W')$ which simulates $(F,[V])$ from a valid coloring of the $\GNECC_{F,W} graph$.

\begin{lem} 
\label{lem2}
For any BAN $F: \mathbb{B}^{n} \to \mathbb{B}^{n}$ and any block-sequential update schedule $W$, $\kappa(F,W) \leq \lceil \log_2(\chi(\GNECC_{F,W})) \rceil$.
\end{lem}

\begin{proof}
Let $k = \lceil \log_2(\chi(\GNECC)) \rceil $.
 We define $W'$ such that we start by updating sequentially the last $k$ nodes, and after this, we update as $W$: $W' = (\{n\}, \{n+1\}, \dots,\{n+k-1\}, W_0, W_1, \dots , W_ {p-1})$.
 Let $\mathsf{color}: \mathbb{B}^n \to \mathbb{N}$ be a minimum coloring of the $\GNECC$ graph. For all $x \in \mathbb{B}^n$, let $\COLOR(x) $ be the number $\mathsf{color}(x)$ encoded with a Boolean vector of size $k$.
 It is possible to encode it with $k$ Boolean numbers because with $k$ bits we can encode $2^k\ge\chi(\GNECC)=|\mathsf{color}(\mathbb B^n)|$ values.
Let $x \in \mathbb{B}^n$ and $y \in \mathbb{B}^k$. 
We define $z=x||y\in\mathbb B^{n+k}$ by $z_{\co0n}=x$ and $z_{\co n{n+k}}=y$.
 For all $j \in \cc0p$, let $A_j(x||y) = \{F(x')\ |\ x' \in \mathbb{B}^n \text{ and } \COLOR(x') = y \text{ and } F_{W_{<j}}(x') = x\}$. We can prove that $|A_j(x||y)| \leq 1$.
 For the sake of contradiction, suppose $\exists F(x'), F(x'') \in A_j(x||y),\ F(x') \neq F(x'')$. Clearly, $(x',x'') \in \NEC$.
 Moreover, $F_{W_{<j}}(x') = x = F_{W_{<j}}(x'')$ gives that $(x',x'') \in \ICC$. So $(x',x'') \in \NECC$.
 However, $\COLOR(x') = y = \COLOR(x'')$, which contradicts the construction of the coloring. 
Let $F': \mathbb{B}^{n+k} \to \mathbb{B}^{n+k}$ be defined for all $x||y\in \mathbb{B}^{n+k}$ by $F'_{\co n{n+k}}(x||y) = \COLOR(x)$ and $\forall j \in \co0p,\ F'(x||y)_{{W'}_{k+j}} = z_{W_j}$ if $A_j(x||y) = \{z\}$, and $[0]^{|W_j|}$ if $A_j(x||y)$ is empty.
 Now, let $z=x||y \in \mathbb{B}^{n+k}$ and we show that $F'^{W'}(x||y)_{\co0n} = F(x)$.
Let us show by induction that $\forall j\in \cc0p,\ F'^{{W'}^{<k+j}}(z)_{\co0n} = F_{W_{<j}}(x) $.
Let $j=0$. We have $F'^{{W'}^{<k+j}}(z)_{\co0n} = F'^{{W'}^{<k}}(z)_{\co0n} = x$ (because in the first $k$ steps of $W'$ we only update the automata of $\co n{n+k}$) and $F_{W_{<j}}(x) = F_{W_{<0}}(x) = x$.
 So $F'^{{W'}^{k+j}]}(z)_{\co0n} = F_{W_{<j}}(x)$.
 Now let $j \in \cc0p$, $z'=F'^{{W'}^{<k+j}}(z)$, and assume that $z'_{\co0n} = F_{W_{<j}}(x)$.
 We have $F'^{{W'}^{<k+j+1}}(z)_{\co0n} = F'_{{W'}^{k+j+1}}(z')$. Thus, $F'^{{W'}^{<k+j}}(z)_{\co0n \setminus {W'}_{k+j+1}} = z'_{\co0n \setminus {W'}_{k+j+1}} = F_{W_{<j}}(x)_{\co0n \setminus {W}_{j+1}} = F_{W_{<j+1}}(x)_{\co0n \setminus {W}_{j+1}}$.
 Furthermore, $\COLOR(x) = F(z)_{\co n{n+k}} = z'_{\co n{n+k}}$, and by induction hypothesis, $F_{W_{<j}}(x) = z'_{\co0n}$.
 Thus, $F(x) \in A_j(z')$. As a consequence, $F'^{{W'}^{<k+j+1}}(z)_{{W'}_{k+j+1}}=F'_{{W'}^{k+j+1}}(z')_{{W'}_{k+j+1}}$ was defined as $F(x)_{{W'}_{k+j+1}} = F(x)_{{W}_{j+1}}$. As a result, $F'^{{W'}^{<k+j+1}}(z)_{\co0n} = F_{W_{<j+1}}(x)$.
 Consequently, $\forall z = x||y \in \mathbb{B}^{n+k},\ F'^{W'}(z)_{\co0n} = F(x)$. Thus, $(F',W') \vartriangleright^{id} (F, [V] )$. Finally, $\kappa(F,W) \leq \lceil \log_2(\chi(\GNECC)) \rceil$.
\end{proof}

Lemma~\ref{lem1} and Lemma~\ref{lem2} show that there is an equivalence between a coloring of the $\GNECC_{F,W}$ graph and a SBAN $(F',W')$ which simulates $(F,[V])$. 
Moreover, we can see in Lemma~\ref{lem2} that one optimal simulation is always achieved by applying sequentially the additional automata before applying the constrained schedule.

\begin{theorem} 
\label{theorem_chro}
For any BAN $F:\mathbb{B}^{n} \to \mathbb{B}^{n}$ and any block-sequential update schedule $W$, $\kappa(F,W) = \lceil \log_2(\chi(\GNECC_{F,W})) \rceil$.
\end{theorem}

In Lemma~\ref{lem_kp_lower_bound} below, using the example of $n/2$ automata which exchange their values, we find a lower bound for $\kappa_n$. We use the fact that if we take the good update schedule $W$, this $\GNECC_{F,W} $ graph has a big clique number.

\begin{lem} 
\label{lem_kp_lower_bound}
$\forall n \in \mathbb{N},\ \kappa_n \geq \lfloor n/2 \rfloor$.
\end{lem}

\begin{proof}
Let us suppose that $n$ is even (if not, we just have to add a useless automaton and the proof remains valid). Let us consider the BAN $F$ such that: 
\begin{equation*}
	\forall i \in \co0{n/2},\ f_i(x) = x_{i+n/2}\quad \text{and}\quad 	
	\forall i \in \co{n/2}n,\ f_i(x) = x_{i-n/2}\text{.}
\end{equation*}
We also consider the simple sequential update schedule $W = (\{0\}, \dots, \{n\})$. Let $X = \{ x \in \mathbb{B}^n\ |\ x_{\co{n/2}n} = [0]^{n/2} \}$, and $x, x' \in X$ such that $x \neq x'$. When we update the first half of the automata, $x$ and $x'$ both become the configuration full of $0$. Then, for $i = n/2$, we have $F_{W_{<i}}(x) = [0]^{n} = F_{W_{<i}}(x')$. Thus, $(x,x') \in \CC$. We also have $x \neq x'$. So $\exists i \in \co{n/2}n$ such that $x_i \neq x'_i$ and $f_{i+n/2}(x) = x_i$ and $f_{i+n/2}(x') = x'_i$. Consequently, $f_{i+n/2}(x) \neq f_{i+n/2}(x')$. Then, $F(x) \neq F(x')$ and $(x,x') \in \NEC$. As a result, we have $(x,x') \in \NECC$. We know that $X$ is a clique. Moreover, $X$ is a clique of size $2^{n/2}$. Thus, the chromatic number of the $\GNECC$ graph is at least $2^{n/2}$ and $\kappa(F,W) \geq n/2$. Hence, $\forall n \in \mathbb{N},\ \kappa_n  \geq n/2$.
\end{proof}

We conjecture that $\lfloor n/2 \rfloor$ is the upper bound as well. This conjecture has not been proven yet, but Theorem~\ref{th_clique} supports it by giving an upper bound to the clique number of a $\GNECC$ graph.

\begin{conjecture} 
\label{conjecture1}
$\forall n \in \mathbb{N},\ \kappa_n \leq \lfloor n/2 \rfloor $.
\end{conjecture}

\section{$\GINECC$ graph} \label{section_inecc}

In this section, we define the $\GINECC$ graph which is the $\GNECC$ graph after we quotient its configurations which have the same image. We can prove that the $\GINECC$ graph has a bigger chromatic number than the $\GNECC$ graph, find an upper bound of its chromatic number and deduce an upper bound for the $\GNECC$ graph as well.

\begin{definition}
The \emph{$\GINECC$ graph} is the graph such that:
\begin{itemize}
\item the vertex set is $\{F(x)\ |\ x \in \mathbb{B}^n \}$, \ie the set of the images of the configurations of the $\GNECC$ graph;
\item two vertices $y$ and $y'$ are connected to each other if $\exists x, x' \in \mathbb{B}^n$ such that $F(x) = y, F(x') = y'$ and $(x,x') \in \NECC$.
\end{itemize}
\end{definition}

Let us now prove that we can use a valid coloring of the $\GINECC$ graph to color the $\GNECC$ graph.

\begin{lem} 
\label{lem_inecc1}
$\chi(\GINECC) \geq \chi(\GNECC)$.
\end{lem}

\begin{proof}
We partition the configurations into sets of equivalent configurations (\ie configurations which have the same image) $E_1, E_2, \dots, E_k$. We denote by $y^i \in \mathbb{B}^n$ the image of the configurations of $E_i$ for each $i \in \co0k$. In other words, $\forall i \in \co0k, \forall x \in E_i,\ F(x) = y^i$. Let $\mathsf{color}: \co0k \to \mathbb{N}^*$ be an optimal coloring of the $\GINECC$ graph. In the $\GNECC$ graph, we can color all the configurations of a set $E_i$ by the color of $y^i$ in the $\GINECC$ graph. Let $x,x' \in \mathbb{B}^n$. If $x$ and $x'$ have the same color:
\begin{itemize}
\item either $x$ and $x'$ are in the same set $E_i$, and then $(x,x') \notin{} \NECC$ 
	because they are equivalent;
\item or they are in two distinct sets $E_i$ and $E_{i'}$. In this case $(x,x') \notin{} 
	\NECC$ otherwise $y^i$ and $y^{i'}$ would be connected in the $\GINECC$ graph and they 
	would have different colors.
\end{itemize}
So, the coloring is a valid coloring and does not need more colors than the $\GINECC$ graph coloring and we conclude that $\chi(\GINECC) \geq \chi(\GNECC)$.
\end{proof}

\begin{remark}
We can see that if we take two SBANs $(F,W)$ and $(F,W')$ with $W'$ a sequentialized version of $W$ (\ie an update schedule that breaks the blocks of $W$ into blocks of size $1$), the chromatic number of the $\GNECC$ graph of $(F,W)$ is always greater than or equal to that of the $\GNECC$ graph of $(F,W')$. Indeed, the set of edges of the $\GNECC$ graph of $(F,W)$ is included in the set of edges of the $\GNECC$ graph of $(F,W')$. Thus, the chromatic number of the latter is greater. Furthermore, the same reasoning applies to the $\GINECC$ graph. \emph{As a result, if we want to find an upper bound to the chromatic number of the $\GNECC$ or $\GINECC$ graph, we can restrict our study to SBAN updated sequentially.}
\end{remark}

\begin{remark}
We can see that if we have a SBAN $(F,W)$, with $W$ a sequential update schedule, we can find another SBAN $(F',W')$ with $W'$ the simple sequential update schedule $(\{0\}, \{1\}, \cdots, \{n-1\})$ which will have the same $\GNECC$ and $\GINECC$ graphs up to a permutation. As a consequence, their chromatic numbers of their $\GNECC$ and $\GINECC$ graphs are equal, respectively.
Thus, if we want to find an upper bound to the chromatic number of the $\GNECC$ or $\GINECC$ graph, we can restrict our study to the SBAN with the simple sequential update schedule $(\{0\}, \{1\}, \cdots, \{n-1\})$.
\end{remark}

Let us find now an upper bound for the chromatic number of the $\GINECC$ graph, by defining a coloring method of the graph based on a greedy algorithm. 

\begin{lem} 
\label{lem_inecc2}
 	$\chi(\INECC) \leq 2^{2n/3+2}$.
\end{lem}

\begin{proof}
Consider the BAN $F: \mathbb{B}^n \to \mathbb{B}^n$ and the simple sequential update schedule $W=(\{0\},\{1\},\cdots,\{n-1\})$. We partition the configurations into sets of equivalent configurations $E_1, E_2, \dots, E_k$. Let us denote by $y^i \in \mathbb{B}^n$ the images of the configurations of $E_i$ for each $ i \in \cc1k$. In other words, $ \forall i \in \cc1k, \forall x \in E_i,\ F(x) = y^i$. We denote the neighbors of the $i^{th}$ image by $N(i)$, \ie
\begin{equation*}
	N(i) = \{i'\ |\ \exists x \in E_i, x' \in E_{i'},\ (x,x') \in \NECC\}\text{.}
\end{equation*}
The degree of the $i^{th}$ image is denoted by $D(i) = |N(i)|$. We sort the images by decreasing degree so that $\forall i<i'$, $D(i) \geq D(i')$.  To choose the color of $y^i$, we apply a greedy algorithm. We use the smallest color not already used by a neighbor of $y^i$: $\mathsf{color}(y^i) = min(\mathbb{N}^* \setminus \{\mathrm{color}(y^{i'})\ |\ i' < i \text{ and } i' \in N(i)\})$.

We can see that it is a proper coloring. Let us prove that if $(y^i,y^{i'}) \in \INECC$ then $\mathsf{color}(y^i) \ne \mathsf{color}(y^{i'})$. Indeed, let $(y^i, y^{i'}) \in \INECC$.
With no loss of generality, let us say that $i' < i$. By definition of $\INECC$, $\exists (x,x') \in \NECC$ such that $F(x) = y^i$ and $F(x') = y^{i'}$.
So $ i' \in N(i)$, and by definition of $\mathsf{color}$, $ \mathsf{color}(y^i) \ne  \mathrm{color}(y^{i'})$. As a consequence, that is a proper coloring.

Now, let $c$ be the biggest color used and $k'$ the index of (one of) the images which have $c$ as color.
By construction, we have $c \leq D(E_{k'})+1$ and $c \leq k'$.
For all $i$, we note $\ell_i = \left\lfloor \log_2(D(E_i)+1) \right\rfloor$ and $\ell = \ell_{k'}$. Since $c \leq D(E_{k'})+1$, we have $c \leq 2^{\ell+1}$. Consider $M(i)=\{i'\ |\ (y^{i})_{\cc0{n-\ell_i}} = (y^{i'})_{\cc0{n-\ell_i}}\}$ and $L(i)=N(i)\setminus M(i)$.
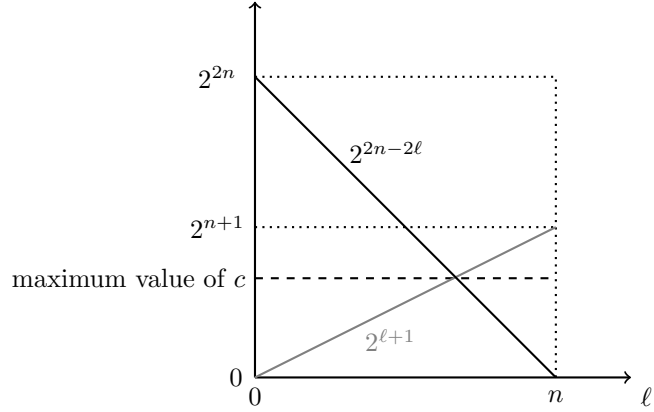
\begin{figure}[t!]
	\begin{center}
    	\begin{tikzpicture}
			\draw[thick,->] (0,0) -- (5,0) node[anchor=north west] {$\ell$};%
			\draw[thick,->] (0,0) -- (0,5) node[anchor=south east] {};%
			\node[color=black] at (1.75,3) {$2^{2n-2\ell}$};%
			\node[color=gray] at (1.8,0.5) {$2^{\ell+1}$};%
			\draw[thick,dotted] (4,0) -- (4,4) node[anchor=south east] {};%
			\node at (4,-0.25) {$n$};%
			\node at (0,-0.25) {$0$};%
			\node at (-0.25,0) {$0$};%
			\node at (-0.5,4) {$2^{2n}$};%
			\node at (-0.5,2) {$2^{n+1}$};%
			\draw[thick,dotted] (0,4) -- (4,4) node[anchor=south east] {};%
			\draw[thick,dotted] (0,2) -- (4,2) node[anchor=south east] {};%
			\draw[thick,dashed] (0,1.32) node[left] {maximum value of $c$} -- (4,1.32) 
				node[anchor=south east]{};%
			\draw[thick,color=gray] (0,0) -- (4,2) node[anchor=north west]{};%
			\draw[thick,color=black] (0,4) -- (4,0) node[anchor=north west]{};%
		\end{tikzpicture}
	\end{center}
	\caption{Upper bound for $c$.}
	\label{fig:upperc}
\end{figure}
Clearly, $|M(i)
| \leq 2^{\ell_i-1}$, and $i \in M(i)
$. So $L(i) = (N(i) \cup \{ i \}) \setminus M(i)
$.
We also know that $i \notin{} N(i)$. As a consequence, $|N(i) \cup \{ i \}| = D(E_i) +1 \geq 2^{\ell_i}$.
Thus, $|L(i)| \geq 2^{\ell_i} - 2^{\ell_i-1}=2^{\ell_i-1}$.

Moreover, $\forall x \in E_i,\ \{x' \in E_{i'}\ |\ i' \in L(i) \text{ and } (x,x') \in \NECC\} \subseteq \{x'\ |\ x_{\oo{n-\ell_i}n} = x'_{\oo{n-\ell_i}n}\}$ because such a pair $(x,x')$ should be confusable at some step $j\le n-\ell_i$.
So $\forall x \in E_i,\ |\{ x' \in E_{i'}\ |\ i' \in L(i) \text{ and } (x,x') \in \NECC \}| \leq 2^{n-\ell_i+1}$.
Putting things together, we get:
\begin{eqnarray*}
2^{\ell_i-1}&\ge&|L(i)|\\
&\ge&|\{(x,x') \in E_i\times E_{i'}\cap\NECC |\ i' \in L(i) \}|\\
&\ge&|E_i|2^{n-\ell_i+1}.
\end{eqnarray*}
We get $|E_i|\le2^{\ell_i-1}/2^{n-\ell_i+1}=2^{2\ell_i-n-2}$.

Furthermore, $\sum \limits_{i=1}^{k'} |E_i| \leq 2^n $ and $\forall i \leq k',\ |E_i| \geq 2^{2\ell_i-n-2} \geq 2^{2\ell-n-2}$.
So $k' 2^{2\ell-n-2} \leq 2^n$ and $k' \leq 2^{2n+2-2\ell}$.
Thus, $c \leq 2^{2n+2-2\ell}$.
However, we have also $c \leq 2^{\ell+1}$.
An upper bound for $c$ is reached when $2^{\ell+1} = 2^{2n+2-2\ell}$ (see Figure~\ref{fig:upperc}).
In other words, when $2^{3\ell} = 2^{2n+1} \iff 2^\ell = 2^{(2n+1)/3}$. So, we have $c \leq 2^{(2n+1)/3+1}$ and $c \leq 2^{2n/3+2}$. Furthermore, $\chi(\GINECC) \leq c$. As a result, $\chi(\GINECC) \leq 2^{2n/3+2}$.
\end{proof}

From Lemma~\ref{lem_inecc1} and Lemma~\ref{lem_inecc2}, we can deduce an upper bound for the chromatic number of a $\GNECC$ graph. Furthermore, using the relation between the chromatic number of a $\GNECC_{F,W}$ graph and $\kappa(F,W)$, we can find an upper bound for $\kappa_n$.
\begin{theorem}\label{thm2}
	$\forall n \in \mathbb{N}, \kappa_n \leq 2n/3+2$.
\end{theorem}

\begin{proof} 
Let $F: \mathbb{B}^n \to \mathbb{B}^n$ and $W \in \overrightarrow{\mathscr{P}}(V)$.
Thanks to Lemma~\ref{lem_inecc1} and Lemma~\ref{lem_inecc2}, we know that $\chi(\GNECC_{F,W}) \leq \chi(\GINECC_{F,W})$ and $\chi(\GINECC_{F,W}) \leq 2^{2n/3+2}$. As a consequence, $\chi(\GNECC_{F,W}) \leq 2^{2n/3+2}$, $\log_2(\chi(\GNECC_{F,W})) \leq 2n/3+2$ and $\kappa(F,W) \leq 2n/3+2$. Thus, we have $\forall F: \mathbb{B}^n \to \mathbb{B}^n$ and $W \in \overrightarrow{\mathscr{P}}(V)$, $\kappa(F,W) \leq 2n/3+2$, which gives by definition, $\kappa_n \leq 2n/3+2$.
\end{proof}

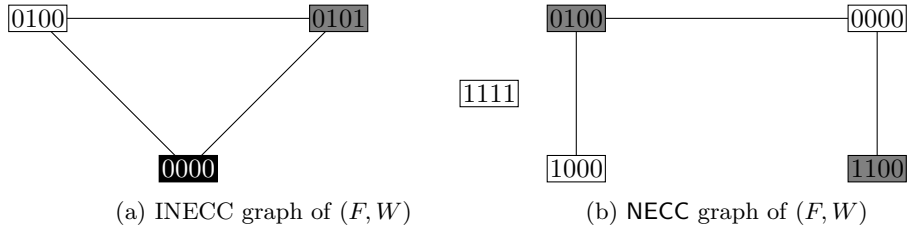
\begin{figure}[b!]
	\begin{center}
		\begin{subfigure}[t]{0.49\textwidth}
			\centerline{
		    \begin{tikzpicture}
			    \node [draw,outer sep=0,inner sep=1,minimum size=10,fill=white] (v0100) at 
			    	(0,0) {0100};
				\node [draw,outer sep=0,inner sep=1,minimum size=10,fill=gray] (v0101) at 
					(4,0) {0101};
				\node [draw,outer sep=0,inner sep=1,minimum size=10,text=white,fill=black] 
					(v0000) at (2,-2) {0000};
				\node [draw,outer sep=0,inner sep=1,minimum size=10,fill=white] (v1111) at 
					(6,-1) {1111};
				\draw [-] (v0100) -- (v0101);
				\draw [-] (v0101) -- (v0000);
				\draw [-] (v0000) -- (v0100);
			\end{tikzpicture}
			}
		    \caption{$\INECC$ graph of $(F,W)$}
			\label{figure_inecc_graph}
		\end{subfigure}%
		~ 
		\begin{subfigure}[t]{0.5\textwidth}
			\centerline{
			\begin{tikzpicture}
				\node [draw,outer sep=0,inner sep=1,minimum size=10,fill=white] (v0000) at 
					(0,-2) {1000};
				\node [draw,outer sep=0,inner sep=1,minimum size=10,fill=gray] (v1000) at 
					(0,0) {0100};
				\node [draw,outer sep=0,inner sep=1,minimum size=10,fill=white] (v0100) at 
					(4,0) {0000};
				\node [draw,outer sep=0,inner sep=1,minimum size=10,fill=gray] (v1100) at 
					(4,-2) {1100};
				\draw [-] (v0000) -- (v1000);
				\draw [-] (v1000) -- (v0100);
				\draw [-] (v0100) -- (v1100);
			\end{tikzpicture}
			}
			\caption{$\GNECC$ graph of $(F,W)$}
			\label{figure_necc_graph}
		\end{subfigure}
		\caption{$\INECC$ and $\GNECC$ graphs of $(F,W)$.}
	\end{center}
\end{figure}

\begin{remark}
The chromatic number of the $\GINECC$ graph gives an upper bound for the $\GNECC$ graph. However, the $\GNECC$ graph can have a smaller chromatic number. For instance, let us consider the following BAN. Let $F: \mathbb{B}^4 \to \mathbb{B}^4$ be such that $F((0,0,0,0)) = (0,0,0,0)$, $F((1,1,0,0)) = (0,0,0,0)$, $F((1,0,0,0)) = (0,1,0,0)$, $F((0,1,0,0)) = (0,1,0,1)$, and for all other $x \in \mathbb{B}^4, F(x) = (1,1,1,1)$.
Let $W$ be the simple sequential schedule $(\{0\}, \{1\}, \{2\}, \{3\})$. Figures~\ref{figure_inecc_graph} and~\ref{figure_necc_graph} show that the chromatic number of the $\GINECC$ and $\GNECC$ graphs are respectively $3$ and $2$. So, even if the worst $\GINECC$ graph had a chromatic number equal to $2^{2n/3}$, it would not disprove the conjecture: we can still hope that the worst $\GNECC$ graph has a better chromatic number, by coloring some equivalent configurations differently.
\end{remark}

\section{Clique number in the $\GNECC$ graph} \label{section_clique_number}

The \emph{clique number} of a graph $G$, denoted by $\omega(G)$, is the size of the biggest clique of $G$. We denote by $\omega(\GNECC)$ the clique number of the $\GNECC$ graph. In this part, we find the maximum value that $\omega(\GNECC)$ can get.
It is important because we know that the chromatic number is bigger that the clique number.
So if in a $\GNECC$ graph the clique number were bigger than $2^{n/2}$, then the chromatic number would be bigger as well and the conjecture would be wrong.
However, if the clique number is smaller than $2^{n/2}$, then we cannot deduce anything about the conjecture. Lemma~\ref{lemsteps} below proves that the set of steps at which two configurations are confusable is an interval.

\begin{lem}
\label{lemsteps}
Let $(x,x') \in \CC$, $I = \ICC(x,x')$, $a = min(I)$ and $b = max(I)$.
Then $I = \cc ab$.
\end{lem}

\begin{proof}
Since $a = min(I)$ and $b = max(I)$, we have $I \subseteq \cc ab$.
For the sake of contradiction, let us suppose that there exists $j \in \cc ab$ such that $j \not \in I$. Let $j$ be the smallest such number. So $F_{W_{<j}}(x) \neq F_{W_{<j}}(x')$, $j \neq a$ because $a \in I$ and $j-1 \in \cc ab$ (because  ${j\neq a}$). Furthermore, $j-1$ does not valid this propriety, because $j$ is the smallest number which validates it.As a consequence, $F_{W_{<j-1}}(x) = F_{W_{<j-1}}(x')$ and $ F_{W_{<j}}(x) \neq F_{W_{<j}}(x')$. So $F(x)_{W_{j-1}} \neq F(x')_{W_{j-1}}$. Furthermore, $F_{W_{<b}}(x)_{W_{j-1}} = F(x)_{W_{j-1}}$ because $j \leq b$ (and then $W_{j-1} \subseteq W_{<b}$) and $F_{W_{<b}}(x')_{W_{j-1}} = F(x')_{W_{j-1}}$.
So $ F_{W_{<b}}(x)_{W_{j-1}} \ne F_{W_{<b}}(x')_{W_{j-1}}$, and thus $ F_{W_{<b}}(x) \neq F_{W_{<b}}(x')$.
As a consequence, $b \notin{} I$ which is a contradiction.
This gives $I = \cc ab$.
\end{proof}

Lemma~\ref{lem10} shows that if two configurations are confusable with a third one at a given step, then they are also confusable between themselves at this step.

\begin{lem} 
\label{lem10}
Let $x,x',x'' \in \mathbb{B}^n$. We have: $\ICC(x,x') \cap \ICC(x,x'') \subseteq \ICC(x',x'')$.
\end{lem}

\begin{proof}
Let $i \in \ICC(x,x') \cap \ICC(x,x'')$. 
Thus, $F_{W_{<i}}(x) = F_{W_{<i}}(x')$ and $F_{W_{<i}}(x) = F_{W_{<i}}(x'')$. As a consequence, $F_{W_{<i}}(x') =  F_{W_{<i}}(x'')$ and $i \in \ICC(x',x'')$. Hence, we have $\forall i \in \ICC(x,x') \cap \ICC(x,x''), i \in \ICC(x',x'')$. 
\end{proof}

Lemma~\ref{lem11} shows that if two configurations are confusable with a third one, the two former ones are confusable if and only if they are confusable with the third one at some simultaneous step.

\begin{lem} 
\label{lem11}
Let $x,x',x'' \in \mathbb{B}^n$ such that $ (x,x') \in \CC$ and $(x,x'') \in \CC$. Then, we have: $\ICC(x,x') \cap \ICC(x,x'') \neq \emptyset \iff (x',x'') \in \CC$.
\end{lem}

\begin{proof}
Suppose that $\ICC(x,x') \cap \ICC(x,x'') \neq \emptyset$.
By Lemma~\ref{lem10}, we know that $\ICC(x,x') \cap \ICC(x,x'') \subseteq \ICC(x',x'')$.
So $\ICC(x',x'') \neq \emptyset$.
As a result, $(x',x'') \in \CC$.
Now, suppose that we have $(x',x'') \in \CC$ and let $ \cc ab = \ICC(x,x')$ and $ \cc{a'}{b'} = \ICC(x,x'')$.
For the sake of contradiction, consider that $\ICC(x,x') \cap \ICC(x,x'') = \emptyset$, \ie $\cc ab \cap \cc{a'}{b'}  = \emptyset$.
 With no loss of generality, consider that 
 $0 \leq a \leq b < a' \leq b' 
 <p = |W|$.
Let $j \in \ICC(x',x'')$. Thus, $F_{W_{<j}}(x') = F_{W_{<j}}(x'')$.
 We can show that $j \not \in \cc ab\cup \cc{a'}{b'}$.
Indeed, if $j \in \cc ab$, then $j \in \ICC(x,x')$ and $F_{W_{<j}}(x) = F_{W_{<j}}(x')$.
So $F_{W_{<j}}(x) = F_{W_{<j}}(x'')$ (because, by definition of $j$, we have $F_{W_{<j}}(x') = F_{W_{<j}}(x'')$) and, as a consequence, $j \in \ICC(x,x'')$ and thus $j \in \ICC(x,x') \cap \ICC(x,x'')$.
As a result, $\ICC(x,x') \cap \ICC(x,x'') \neq \emptyset$.
 There is a contradiction, so $j \not \in \cc ab$. Similarly, we can prove that $j \notin{} \cc{a'}{b'}$.
Now, let us prove that $j \notin{} \co{0}{a }$.
For the sake of contradiction let us say that $j \in \co{0}{a }$.
Then, $\exists j' \in \oo{j}{a },$ $F(x')_{W_{j'}} \neq F(x'')_{W_{j'}}$.
Otherwise, we would have $F_{W_{<a}}(x'') = F_{W_{<a}}(x') = F_{W_{<a}}(x)$ and then $\cc ab \cap \cc{a'}{b'}  \neq \emptyset$.
Furthermore, we know that $F_{W_{<a}}(x') = F_{W_{<a}}(x)$ (because $a \in \ICC(x,x')$) and $W_{j'} \subseteq W_{<a}$ (because $j' < a$) so $F(x')_{W_{j'}} = F(x)_{W_{j'}}$ and thus $F(x'')_{W_{j'}} \neq F(x)_{W_{j'}}$.
 As a consequence, $F_{W_{<a'}}(x')_{W_{j'}} \neq F_{W_{<a'}}(x)_{W_{j'}}$ (because $W_{j'} \subseteq W_{<a'}$ since $j' < a < a'$). So $a' \notin{} \ICC(x,x'')$.
This is a contradiction. So $j \notin{} \co{0}{a }$.
Now, let us prove that $j \notin{} \oo{b}{a' } \cup \oo{b'}{p }$.
If $j \in \oo{b}{a' } \cup \oo{b'}{p }$ then $j>b$.
 We know that $F(x)_{W_{b}} \neq F(x')_{W_{b}}$ (otherwise we would have $F_{W_{<b+1}}(x) \neq F_{W_{<b+1}}(x')$ and then $b+1 \in \ICC(x,x')$).
However, we have $F(x)_{W_{b}} = F(x'')_{W_{b}}$, because $W_{b} \subseteq W_{<a'}$ since $b < a'$. So $F(x')_{W_{b}} \ne F(x'')_{W_{b}}$.
 Thus, $F_{W_{<j}}(x') \neq F_{W_{<j}}(x'')$ because ${W_{<j}}$ because $b<j$, which is a contradiction. As a consequence, $j \notin{} \oo{b}{a' } \cup \oo{b'}{p }$.
 As a result, $j$ does not exist. Thus, $ \ICC(x',x'') = \emptyset$, and finally, $ (x',x'') \notin{} \CC $.
\end{proof}

Lemma~\ref{lem12} shows that all cliques of the $\GNECC$ graph have at least one step during which all the configurations of the clique are simultaneously confusable.

\begin{lem}
\label{lem12}
Let $X$ be a clique of the $\GNECC$ graph. Then, we have: $\exists i, \forall x,x' \in X,\  i \in \ICC(x,x')$.
\end{lem}

\begin{proof}
Let $x \in X = \{x^1, x^2, \dots, x^k\}$ such that $|X| = k$, and let $I = I_1 \cap I_2 \cap \dots \cap I_k$ where $I_1 = \ICC(x, x^1), \dots, I_k = \ICC(x, x^k)$. We can prove that all the intervals intersect each other two by two. In other words, $\forall i,i' \in \co0k,\ I_i \cap I_{i'} \neq \emptyset$. For the sake of contradiction, assume that there are disjoint intervals. In this case, we would have $x', x'' \in X$ such that $\ICC(x,x') \cap \ICC(x,x'') = \emptyset$. By Lemma~\ref{lem11}, we would have $(x',x'') \notin \CC$. However, $x',x'' \in X$, so $(x',x'') \in \CC$. There is a contradiction. Consequently, all the intervals intersect each other two by two, and we know that if a set of intervals intersect each other two by two then they have an interval in common. So $I \neq \emptyset$.\\ Let $i \in I$. Now, let us prove that $\forall x', x'' \in X,\ i \in \ICC(x',x'')$. Let $x',x'' \in X$. We have $i \in \ICC(x,x')$ and $i \in \ICC(x,x'')$. Thus, $F_{W_{<i}}(x) = F_{W_{<i}}(x')$ and $F_{W_{<i}}(x) = F_{W_{<i}}(x'')$, which implies that $F_{W_{<i}}(x') = F_{W_{<i}}(x'')$. As a result, $i \in \ICC(x',x'')$ and $\forall x',x'' \in X,\ i \in \ICC(x',x'')$.
\end{proof}

Using Lemma~\ref{lem12}, Theorem~\ref{th_clique} shows that the clique number of any $\GNECC$ graph is less than or equal to $2^{n/2}$.

\begin{theorem}
\label{th_clique}
$\omega(\GNECC) \leq 2^{ \left\lfloor n/2 \right\rfloor}$.
\end{theorem}

\begin{proof}
Let $X$ be the biggest clique of the $\GNECC$ graph, $x \in X$ and $i$ such that $\forall x, x' \in X,\ i \in \ICC(x,x')$ (Thanks to Lemma~\ref{lem12}, we know there is one). In other words, $\forall x' \in X,\ F_{W_{<i}}(x') = F_{W_{<i}}(x)$. So $\forall x,x' \in X,\  x_{\overline{W_{<i}}} = x'_{\overline{W_{<i}}}$ and $F(x)_{W_{<i}} = F(x')_{W_{<i}} $. Let $x \in X$. There are 2 cases:
\begin{itemize}
\item $|W_{<i}| < n/2$. Then, we have $|\overline{W_{<i}}| \geq n/2 $. Thus, $|\{x'\ |\  x'_{\overline{W_{<i}}} = x_{\overline{W_{<i}}}\}| < 2^{n/2}$ and, since $X \subseteq \{x'\ |\ x'_{W_{<i}} = x_{W_{<i}} \}$, we have $|X| < 2^{n/2}$.
\item $|W_{<i}| \geq n/2$. Then, we have $\{F(x')\ |\ x' \in X \} \subseteq  \{x'\ |\  F(x')_{W_{<i}} = F(x)_{W_{<i}}\} $ and $|\{F(x')\ |\ F(x')_{W_{<i}} = F(x)_{W_{<i}}\}| \leq 2^{n/2}$.  In this case, since all configurations of $X$ are not equivalent, we have $\forall x, x' \in X,\ x \neq x' \implies F(x) \neq F(x')$. Thus, $|X| \leq |\{F(x')\ |\  x' \in X \}|$. As a consequence, $|X| \leq 2^{n/2}$.
\end{itemize}
In all cases, we have $|X| \leq 2^{n/2}$. So $\omega(\GNECC) \leq 2^{n/2}$.
\end{proof}

This result supports Conjecture~\ref{conjecture1} because the $\GNECC$ graphs with the biggest chromatic number that we succeeded to build are graphs with big clique number. It seems we reached the limit of this technique.

\section{Class of bijective BANs} \label{section_special_classes}

In this part, we study BANs whose global transition functions are bijective, \ie BANs whose dynamics with a parallel update schedule are only composed of recurrent configurations. For this class of BANs, we can prove a result which is really close to the conjecture. We prove this using two intermediate lemmas. The first one is that if two configurations are confusable then either the first parts of the two images are equal or the second parts of the two configurations are.

\begin{lem}
\label{lem13}
If $W = (0, 1, \dots, n)$ then $\forall (x,x') \in \CC,\ F(x)_{\co{0}{n/2 }} = F(x')_{\co{0}{n/2 }}$ or $x_{\co{n/2}{n }} = x'_{\co{n/2}{n }}$.
\end{lem}

\begin{proof}
Let $(x,x') \in CC$. Then, $\exists i \in \co{0}{n },\ F_{\co{0}{i }}(x) = F_{\co{0}{i }}(x')$. Let $i$ be the smallest such number. We have: $F(x)_{\co{0}{i }} = F(x')_{\co{0}{i }}$ and $x_{\co{i}{n }} = x'_{\co{i}{n }}$. Then, $i$ can follow the two cases below:
\begin{itemize}
\item $i \leq n/2$. Then, $\co{n/2}{n } \subseteq \co{i}{n }$ and $x_{\co{n/2}{n }} = x'_{\co{n/2 }{n }}$;
\item $i \geq n/2$. Then, $\co{0}{n/2 } \subseteq \co{0}{i }$ and $F(x)_{\co{0}{n/2 }} = F(x')_{\co{0}{n/2 }}$.
\end{itemize}
And we get the expected result.
\end{proof}

The next lemma is a simple consequence of Lemma~\ref{lem13}: if we take the neighbors of a configuration in a $\GNECC$ graph and we take the set of images of these configurations when we apply $F$, then this set has less than $2^{n/2 + 1} - 2$ elements.

\begin{lem}
\label{lem14}
If $W' = (0, 1, \dots, n)$ then $\forall x \in \mathbb{B}^n,\ |\{F(x')\ |\ (x,x') \in \NECC \}| \leq 2^{ n/2 +1} - 2$.
\end{lem}

\begin{proof}
Let $x \in \mathbb{B}^n$. According to Lemma~\ref{lem13}, $\forall x' \in \mathbb{B}^n,\  F(x)_{\co{0}{n/2 }} = F(x')_{\co{0}{n/2 }}$ or $x_{\co{n/2}{n }} = x'_{\co{n/2}{n }}$.
Then, $\{ x'\ |\ (x,x') \in \NECC\} \subseteq \{x'\ |\ x_{\co{n/2}{n }} = x'_{\co{n/2}{n }}\} \cup \{x'\ |\ F(x)_{\co{0}{n/2 }} = F(x')_{\co{0}{n/2 }}\}$.
So $\{F(x')\ |\ (x,x') \in \NECC \} \subseteq \{F(x')\ |\ x_{\co{n/2}{n }} = x'_{\co{n/2}{n }}\} \cup \{F(x')\ |\ F(x)_{\co{0}{n/2 }} = F(x')_{\co{0}{n/2 }}\}$.
Thus, $|\{F(x')\ |\ (x,x') \in \NECC\}| \leq |\{F(x')\ |\ x_{\co{n/2}{n }} = x'_{\co{n/2}{n }} \}| + |\{F(x')\ |\ F(x)_{\co{0}{n/2 }} = F(x')_{\co{0}{n/2 }}\}|$.
And we have: $|\{ F(x')\ |\ F(x)_{\co{0}{n/2 }} = F(x')_{\co{0}{n/2 }}\}| \leq 2^{n/2}$.
Furthermore, $|\{x'\ |\  x_{\co{n/2}{n }} = x'_{\co{n/2}{n }}\}| \leq 2^{n/2}$.
As a consequence, $|\{F(x')\ |\ x_{\co{n/2}{n }} = x'_{\co{n/2}{n }}\}| \leq 2^{n/2}$.
So, $|\{F(x')\ |\ (x,x') \in \NECC \}| \leq 2^{n/2+1}$.
Furthermore, $F(x) \in \{F(x')\ |\ x_{\co{n/2}{n }} = x'_{\co{n/2}{n }}\}$ and $F(x) \in \{F(x')\ |\  F(x)_{\co{0}{n/2 }} = F(x')_{\co{0}{n/2 }}\}$ but $F(x) \notin{} \{F(x')\ |\ (x,x') \in \NECC \}$.
Consequently, $|\{F(x')\ |\ (x,x') \in \NECC\}| \leq 2^{n/2+1} - 2$, which is the expected result.
\end{proof}

Using the fact that we are talking about a bijective function, and thanks to Lemma~\ref{lem14}, we bound the degree of every configuration in the $\GNECC$ graph. Then, we deduce a bound for the chromatic number of the $\GNECC$ and, thus, a bound for $\kappa$.

\begin{theorem}
If $F: \mathbb{B}^n \to \mathbb{B}^n$ is a bijective function then $\kappa(F,W) \leq n/2+1$.
\end{theorem}

\begin{proof}
Let $F: \mathbb{B}^n \to \mathbb{B}^n$ be a bijective function. For all $x \in \mathbb{B}^n$, let $d(x)$ be the degree of $x$ in the $\GNECC$ graph. In other words, $\forall x,\ d(x) = |\{x'\ |\ (x,x') \in \NECC\}|$. Let $x \in \mathbb{B}^n$ be the configuration with maximal degree. We know by Lemma~\ref{lem14} that $|\{ F(x')\ |\  (x,x') \in \NECC\}| \leq 2^{n/2+1} - 2 $. However, since $F$ is a bijective function, we have $|\{ F(x')\ |\ (x,x') \in \NECC \}| = |\{x'\ |\ (x,x') \in \NECC\}|$ and then, $d(x) \leq 2^{n/2+1} - 2$. So, $\chi(\NECC) \leq 2^{n/2+1} - 1$. Thus, $\log_2(\chi(\GNECC)) \leq \dfrac{n}{2}+1$. As a result, $\kappa(F,W) \leq \dfrac{n}{2}+1$.
\end{proof}

\section{Conclusion and future research}

In this article, we were interested in the minimal number $\kappa$ of additional automata that a SBAN associated with a block-sequential update schedule needs to simulate
another given one with a parallel update schedule, in the worst case. The maximum value that $\kappa$ can take for all SBANs of size $n$ is denoted by $\kappa_n$. To answer this question, we introduced the concept of $\GNECC$ graph, a graph built from SBANs. We proved that the $\log$ of the chromatic number of this graph and the $\kappa$ of a SBAN are the same quantity. We achieved to bound $\kappa_n$ in the interval $[ n/2, 2n/3+2]$ and we conjectured that $\kappa_n$ is equal to $n/2$. To support this conjecture, we showed that the maximum clique number that a $\GNECC$ graph can have is equal to $2^{n/2}$. This means that the $\GNECC$ graph of a SBAN which would have a $\kappa$ greater than $n/2$ would have a $\GNECC$ graph with a chromatic number greater than the clique number. Finally, we showed that the conjecture is true (up to one extra automaton) if we restrain to SBANs whose global transition functions are bijective.\smallskip

More work is needed to close the gap ${[n/2, 2n/3+2]}$ left on $\kappa_n$. There is also a related problem where, given a SBAN with a parallel update schedule, we search the number of additional automata needed for a SBAN with any sequential update schedule (\ie{}, we do not impose any order on the update schedule) to simulate the first SBAN. We can see that for some BANs, this number is really smaller than when we impose an order. We can take the example used in Lemma~\ref{lem_kp_lower_bound}. The BAN has $n/2$ pairs of automata that exchange their values. If the mandatory order is to update one automaton only of every pair of automata and then the other we need $n/2$ additional automata. But if the order is free then we can update all the pairs of automata one at a time and do with only one additional automaton using a parity trick. This is a particular BAN and the problem of finding an upper bound in the worst case better than $\kappa_n$ is still open. 

Furthermore, we could study the issue presented in this article with other kinds of update schedules (which update many times each automata for instance) or other kinds of intrinsic simulations (where many automata can represent one simulated automaton for example).

These results could also help to design new SBANs behaving the same way as a given one, with different update schedule, and as small as possible.
Associated with the concept of functional modularity, we could also use them to replace a small functional module with an unexpected behavior in some situations by another module that is more robust to schedule variations.\bigskip

\noindent \emph{Acknowledgements.} This work has been partially supported by the project PACA APEX FRI.


\end{document}